\documentclass[12pt]{article}
\usepackage{graphicx}
\DeclareGraphicsExtensions{.tif}
\usepackage[T1]{fontenc}
\usepackage[utf8]{inputenc}
\usepackage{authblk}

\usepackage[numbers]{natbib}
\usepackage[OT4]{fontenc}
\usepackage[utf8]{inputenc}

\usepackage[numbers]{natbib}
\usepackage[OT4]{fontenc}
\usepackage[utf8]{inputenc}

\newcommand{\Exp}{{\rm I\hspace{-0.8mm}E}}
\newcommand{\Prob}{{\rm I\hspace{-0.8mm}P}}
\newcommand{\Var}{{\bf Var}}

\newcommand{\Corr}{{\bf Corr}}
\newcommand{\iz}{{\rm \rlap Z\kern 2.2pt Z}}

\newcommand{\ind}{{1\hspace{-1mm}{\rm I}}}

\newtheorem{theorem}{Theorem}

\newtheorem{proposition}{Proposition}
\author[1]{Zbigniew Michna\footnote{Corresponding author\\ Email: zbigniew.michna@ue.wroc.pl\\ Tel/fax: +48 71 3680335}}
\author[2]{Peter Nielsen}
\author[2]{Izabela Ewa Nielsen}
\affil[1]{Department of Mathematics and Cybernetics

Wroc{\l}aw University of Economics}
\affil[2]{Department of Mechanical and Manufacturing Engineering Aalborg University}

\title{\bf\LARGE {The impact of stochastic lead times on the bullwhip effect}}
\date{}

\begin{document}

\maketitle

\bibliographystyle{abbrv}

\begin{abstract}
In this article we want to review the research state on the bullwhip effect in supply chains with stochastic lead times and give a contribution to quantifying the bullwhip effect. We analyze the mo-dels 
quantifying the bullwhip effect in supply chains with stochastic lead times and find advantages and disadvantages of their approaches to the bullwhip problem. Using real data we confirm that real lead times are stochastic and can be modeled by a sequence of independent identically distributed random variables. Moreover we modify a model where stochastic lead times and lead time demand forecasting are considered and give
an analytical expression for the bullwhip effect measure which indicates that
the distribution of a lead time and the delay parameter of the lead time demand prediction are the main
factors of the bullwhip phenomenon.
Moreover we analyze a recent paper of Michna and Nielsen \cite{mi:ni:13} adding simulation results.

\vspace{5mm}
{\it Keywords: supply chain, bullwhip effect,
inventory policy, stochastic lead time, demand forecasting, lead time forecasting, lead time demand forecasting}
%\newline
%\vspace{2cm}
%MSC(2010): Primary ; Secondary.
\end{abstract}

\section{Introduction}
Supply chains are networks of firms (supply chain members) which act in order to deliver a product to the end consumer.
Supply chain members are concerned with optimizing their own objectives and this results in a poor
performance of the supply chain.  In other words local optimum policies of members do not result
in a global optimum of the chain and they yield the tendency of replenishment orders to increase in variability as one moves up stream in a supply chain. This effect was first recognized by Forrester \cite{fo:58} in the middle of the twentieth century and the term of bullwhip effect was coined by Procter \& Gamble management.
The bullwhip effect is considered harmful because of its consequences which are (see e.g. Buchmeister et al. \cite{bu:pa:pa:po:08}):
excessive inventory investment, poor customer service levels, lost revenue, reduced productivity, more difficult decision-making, sub-optimal transportation, sub-optimal production etc. This makes it critical to find the root causes of the bullwhip effect and to quantify the increase in order variability at each stage of the supply chain. In the current state of research several main causes of the bullwhip effect are considered (see e.g. Lee et al. \cite{le:pa:wh:97a} and \cite{le:pa:wh:97b}): demand forecasting, non-zero lead time, supply shortage, order batching, price fluctuation and lead time forecasting (see Michna and Nielsen \cite{mi:ni:13}). To decrease the variance amplification in a supply chain (i.e. to reduce the bullwhip effect) we need to identify all factors causing the bullwhip effect and to quantify their impact on the effect.

Many researchers have assumed a deterministic lead time and studied the influence of different methods of demand forecasting on the bullwhip effect such as simple moving average, exponential
smoothing, and minimum-mean-squared-error forecasts when demands are independent identically distributed or constitute integrated moving-average, autoregressive processes or autoregressive-moving averages (see Graves \cite{gr:99}, Lee et al. \cite{le:so:ta:00}, Chen et al. \cite{ch:dr:ry:si:00a} and \cite{ch:dr:ry:si:00b}, Alwan et al. \cite {al:li:ya:03}, Zhang \cite{zh:04} and Duc et al. \cite{du:lu:ki:08}). Moreover they quantify the impact of a deterministic lead time on the bullwhip effect and it follows from their work that the lead time  is one of the major factors influencing the size of the bullwhip effect in a given supply chain. 
Stochastic lead times were intensively investigated in inventory systems see Bagchi et al. \cite{ba:ha:ch:86},
Hariharan and Zipkin \cite{ha:zi:95}, Mohebbi and Posner \cite{mo:po:98}, Sarker and Zangwill \cite{sa:za:91}, Song \cite{so:94a} and \cite{so:94b}, Song and Zipkin
\cite{so:zi:93} and \cite{so:zi:96} and Zipkin \cite{zi:86}. Most of these works consider the so-called exogenous lead times that is they do not depend on the system e.g. the lead times are independent of the orders or the capacity utilization of a supplier. Moreover these articles studied how the variable lead times affect the control parameter, the inventory level or the costs.
One can investigate the so-called endogeneous lead times that depends on the system. This is analyzed in So and Zheng \cite{so:zh:03} showing the impact of endogeneous lead times on the amplification of the order variance and has been done by simulation. 
%We should mention recent articles on the impact of lead times on the bullwhip effect that is Agrawal et al. \cite{ag:se:sh:09} and Li and Liu \cite{li:li:13}. However, the first paper
%does not consider stochastic lead times and the second paper investigates a transition state model with uncertainties in demands, production process, supply chain structure,
%inventory policy implementation and especially vendor order placement lead time delays. They find a maximally allowable vendor order placement lead time delay such that the supply chain system is exponentially stabilizable.
%This approach uses dynamical control systems theory and is not probabilistic (for similar models see the references in Li and Liu \cite{li:li:13}).
Recently the impact of stochastic lead times on the bullwhip effect is intensively investigated.
The main aim of this article is to review papers devoted to stochastic lead times in supply chains in the context of the bullwhip effect especially those which quantify the effect.  Moreover we modify a model where stochastic lead times and lead time demand forecasting are considered. 
In this model we find an analytical expression for the bullwhip effect measure which indicates that
the distribution of a lead time (the probability of the longest lead time and its expectation and variance) and the delay parameter of the lead time demand prediction are the main
factors of the bullwhip phenomenon.

In Tab. \ref{tabpapers} we collect all the main articles in which models on the bullwhip effect with stochastic lead times are provided (except the famous works of Chen et al. \cite{ch:dr:ry:si:00a} and 
\cite{ch:dr:ry:si:00b} where deterministic lead time is considered and some of them analyze the effect using simulation).

\begin{table}[!h]
\begin{center}
   \begin{tabular}{|c|c|c|c|}
     \hline
Article  & Demands & Lead times & Forrecasting  \\
\hline
Chen et al. \cite{ch:dr:ry:si:00a}&	AR(1)	 & deterministic & moving average \\
&  &  & of demands\\
\hline
Chen et al. \cite{ch:dr:ry:si:00b} & AR(1) & deterministic & expon. smoothing \\
&  &  & of demands\\
\hline
Chaharsooghi & deterministic & i.i.d. & --\\
and Heydari \cite{ch:he:10}&  &  & \\
\hline
Chatfield et al. \cite{ch:ki:ha:ha:04} & AR(1) & i.i.d. & moving average \\
&  &  & of lead time demands\\
\hline
Kim et al. \cite{ki:ch:ha:ha:06} & AR(1) & i.i.d. & moving average \\
& & & of lead time demands\\
\hline
Duc et al. \cite{du:lu:ki:08} & AR(1) & i.i.d. & the minimum-mean-\\
& ARMA(1,1) & & squared-error forecast\\
&  &  & of demands\\
\hline
Fioriolli et al. \cite{fi:fo:08} & AR(1) & i.i.d. & moving average\\
&  &  & of demands\\
\hline
Michna & i.i.d. & i.i.d. & moving average\\
and Nielsen \cite{mi:ni:13} &  &  &of demands \\
&  &  & and lead times\\
\hline
Reiner & dependent & i.i.d. & moving average \\
and Fichtinger \cite{re:fi:09} &  &  & of demands  \\
&  &  & and lead times\\
\hline
So & AR(1) & dependent &  the minimum-mean-\\
and Zheng \cite{so:zh:03} &  &  &squared-error forecast \\
&  &  & of demands\\
\hline
\end{tabular}
\end{center}
\caption{Articles on the impact of lead time on the bullwhip effect}\label{tabpapers}
\end{table}

The remainder of the paper is structured as follows. In the next section a discussion of the bullwhip effect is presented along with the common definition of it. The next section presents a brief study of real lead times in a supply chain, documenting their nature. The following section analyzes the current main models of supply chains with stochastic lead times, expanding and modifying some of the results. Finally conclusions and future research opportunities are presented. 

\section{Supply chains and the bullwhip effect}
In recent studies a supply chain is considered as a system of organizations, people, activities, information and resources involved in moving a product or service from suppliers to customers. More precisely a supply chain consists of customers, retailers, warehouses, distribution centers, manufactures, plants, raw material suppliers etc. They are members or stages (echelons) of a given supply chain. A supply chain has (or is assumed to have) a linear order which means that at the bottom there are customers,
above customers there is a retailer, above the retailer there is e.g. a manufacturer and so on.
The linear order is determined by the flow of products which stream down from the supplier through
the manufacturer, warehouse, retailer to the customers. The financial and information flows can 
accompany the flow of products. The simplest supply chain can consist of customers (customers are not regarded as a stage), a retailer and a manufacturer (a supplier). At every stage (except customers) a member of a supply chain possesses a storehouse and uses a certain stock policy (a replenishment policy) in its inventory control to fulfill its customer (a member of the supply chain which is right below) orders in a timely manner.
Commonly used replenishment policies are: the periodic review,  the replenishment interval, the order-up-to level inventory policy (out policy), $(s,S)$ policy, the continuous review, the reorder point (see e.g. Zipkin \cite{zi:00}).
A member of a supply chain observes demands from a stage below and lead times from a stage above.
Based on the previous demands and previous lead times and using a certain stock policy each member of a chain places
an order to its supplier. Thus at every stage one can observe demands from the stage below and replenishment orders sent to the stage above. The phenomenon of the variance amplification in replenishment orders if one moves up in
a supply chain is called bullwhip effect (see Disney and Towill \cite{di:to:03} and Geary et al. \cite{ge:di:to:06} for the definition and historical review). Manson et al. \cite{mu:hu:ro:03} assert: ''When each member of a group tries to maximize his or her benefit without regard to the impact on other members of the group, the overall effectiveness may suffer''.  The bullwhip effect is the key example of a supply chain inefficiency. 

The main measure of the bullwhip effect is the ratio of variances, that is if $q$ is a random variable describing demands (orders) of a member of the supply chain to a member above and $D$ is a random variable responsible for demands of the member below (e.g. $q$ describes orders of a retailer to a manufacturer (supplier) and $D$ shows
customer demands to the retailer) then the measure of performance of the bullwhip effect is the following
$$
BM=\frac{\Var(\mbox{orders})/\Exp(\mbox{orders})}{\Var(\mbox{demands})/\Exp(\mbox{demands})}=\frac{\Var q/\Exp q}{\Var D/\Exp D}\,.
$$ 
Usually in most models $\Exp D= \Exp q$. The value of $BM$ is greater than one in the presence of the bullwhip effect in a supply chain. If $BM$ is equal to one then there is no variance amplification whereas $BM$ smaller than one indicates dampening which means that the orders are smoothed compared to the demands indicating a push rather than pull supply chain. Another very important parameter of the supply chain performance is the measure of the net stock amplification of a given supply chain member. More precisely let $N_S$ be the level of the net stock of a supply chain member (e.g.
a retailer or a supplier) and $D$ be demands observed from its downstream member (customers or
a retailer) then the following measure 
$$
NSM=\frac{\Var({\mbox{net stock}})}{\Var({\mbox{demands}})}=\frac{\Var({N_S})}{\Var{D}}
$$ 
is also considered as a critical performance measure. In many models it is assumed that
the costs are proportional to $\sqrt{\Var(\mbox{orders})}$ and $\sqrt{\Var({N_S})}$.

\section{Establishing real lead time behavior}
Despite lead times widely being considered as one of the main causes of the bullwhip effect, limited literature exists investigating actual lead time behavior (see Tab. \ref{tabpapers}). Most research to date is focused on lead time demand, with an assumption of either constant lead times or lead times that are independent identically distributed  (i.i.d.). To support the assumptions used in references (see Chatfield et al. \cite{ch:ki:ha:ha:04}, Kim et al. \cite{ki:ch:ha:ha:06},  Duc et al. \cite{du:lu:ki:08} and Michna and Nielsen \cite{mi:ni:13}) – i.e. that lead times are i.i.d. - the lead time behavior from a manufacturing company is analyzed as an example in the following. The data used is 6,967 orders for one product varying in quantity ordered over a period of two years (481 work days) in a manufacturing company. On average 14.5 orders are received per day in the period, each order is to an individual customer in the same geographical region. To test whether or not lead times are in fact i.i.d. the following tests are employed:
\begin{enumerate}
\item[1)] autocorrelation (see e.g. Box and Jenkins \cite{bo:je:76}) for independence of the lead times. This is done on the average lead time per day as the individual orders cannot be ordered in time periods smaller than one day. 
\item[2)] Kolmogorov-Smirnov test (see e.g. Conover \cite{co:71}) is applied. The test is a widely used robust estimator for identical distributions \cite{co:71} . The method (as seen in Fig. \ref{test_fig}) relies on comparing samples of lead times and using Kolmogorov-Smirnov test to determine whether not these pairwise samples are identical. In this research a 0.05 significance level is used and the ratio of pairwise comparisons that pass this significance test is the output from the analysis. Different sample sizes are used to determine if the lead times can be assumed to be similarly distributed in smaller time periods, and thus if it is fair to sample previous lead time observations to estimate lead time distribution for planning purposes.   
\end{enumerate}
For a detailed account of the method please refer to Nielsen et al. \cite{ni:14}.
\begin{figure}[!h]
\begin{center}
\includegraphics[width=13cm]{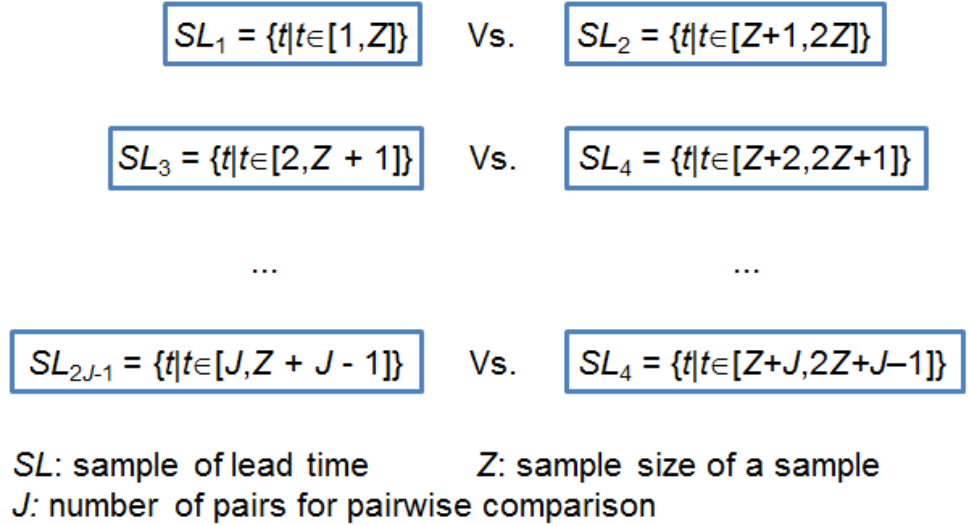}
\caption{Sample and comparison procedure using Kolmogorov-Smirnov test for pairwise comparisons.}\label{test_fig}
\end{center}
\end{figure}
%For a detailed account of the method please refer to Nielsen et al. \cite{ni:14}.

\begin{figure}[!h]
\begin{center}
\includegraphics[height=12cm,width=13cm]{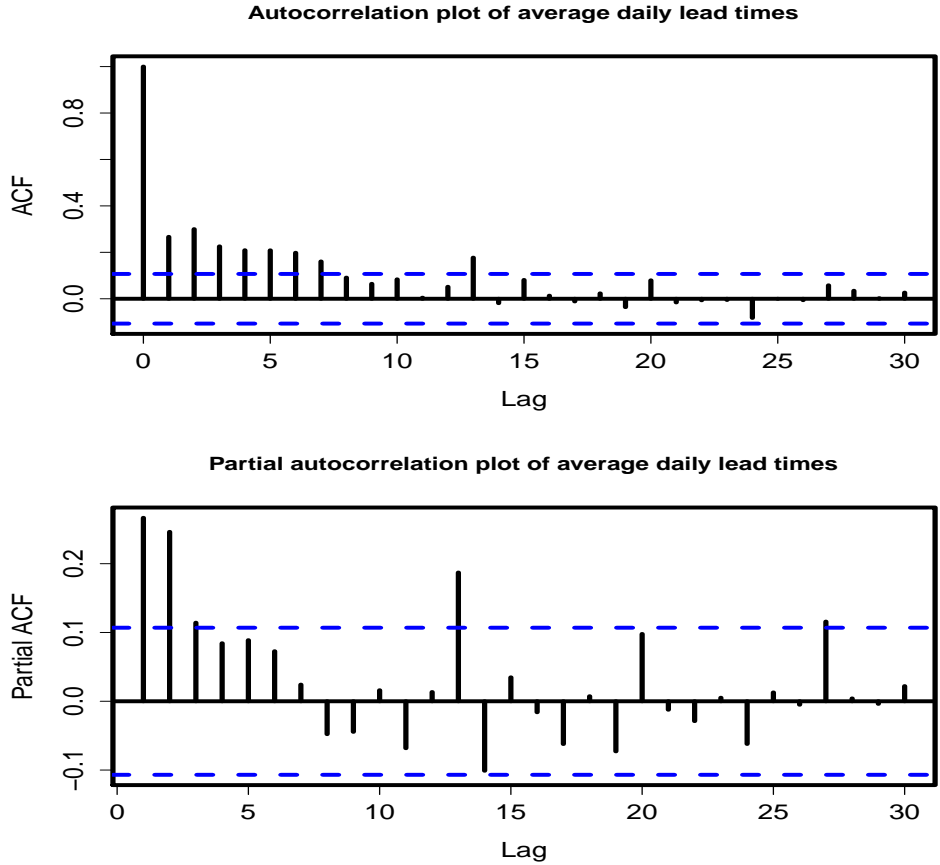}
\caption{top: auto correlation plot of average lead time per day; bottom: partial auto correlation plot of average lead time per day.}\label{iidts}
\end{center}
\end{figure}
An autocorrelation and partial autocorrelation plot are found in Fig. \ref{iidts}
which clearly show that the average lead times per day for all practical purposes can be considered mutually independent. There may be some minor indications from the average lead time on a given day slightly depend on recent average lead times. However, the correlation coefficients are small (that is approximately 0.1) and the penalty for assuming independence seems slight. 

\begin{figure}[!h]
\begin{center}
\includegraphics[height=12cm,width=13cm]{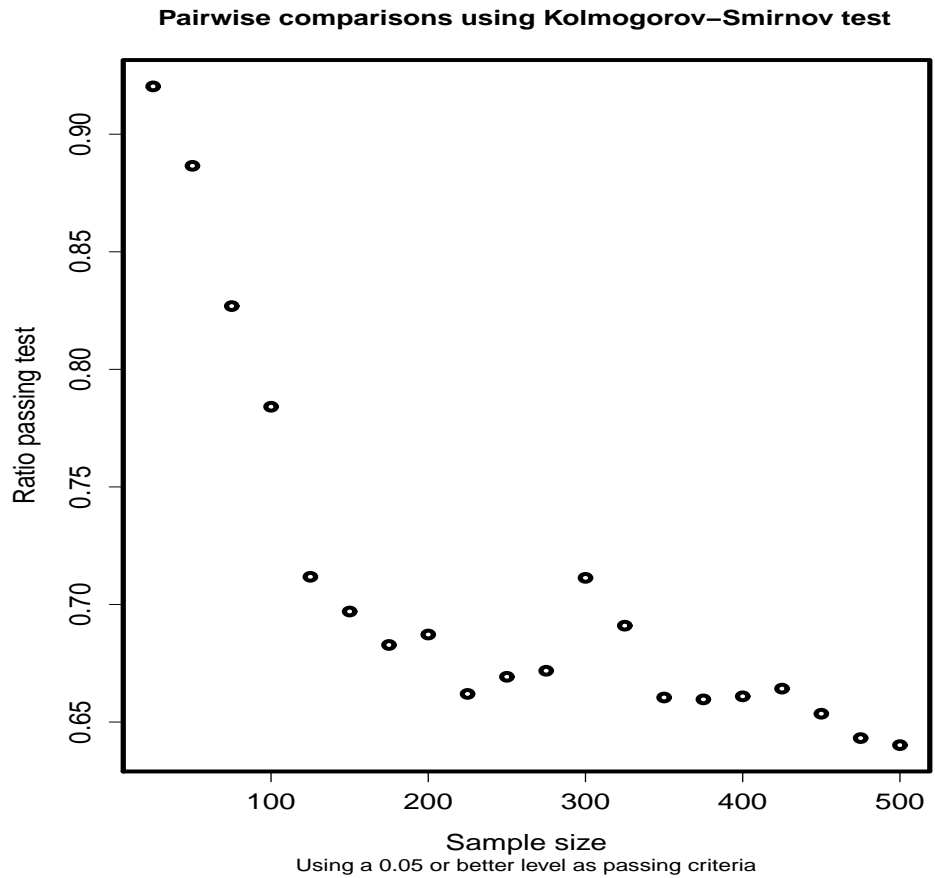}
\caption{Sample and comparison procedure using Kolmogorov-Smirnov test for pairwise comparisons.}\label{fig_ks_samples}
\end{center}
\end{figure}

Fig. \ref{fig_ks_samples} shows that even for large samples (500 vs. 500 observations) most of the comparisons are found to be statistically similar on a 0.05 or better level. This supports the assumptions that the lead times are in fact identically distributed. 

The overall conclusion is that it is not wrong to assume that lead times are in fact i.i.d. The investigation also underlines that it is a grave oversimplification to assume that lead times are constant for individual orders. There is also no guarantee that lead times are in fact i.i.d. in any and all context.

\section{Models with stochastic lead times}
Having established that at least in some cases lead times can be considered to be i.i.d. the next step is to analyze the current state of research into supply chains where lead times are assumed to be stochastic. The lead time is typically regarded as the second main cause of the bullwhip effect after demand forecasting.
Lead times are made of two components which are the physical delays and the information delays.
In models one does not distinguish these components as the lead time is the time between an order is placed by a member of a supply chain and the epoch when the product is delivered to the member.
The assumption that the lead time is constant is rather unrealistic. Undoubtedly in many supply chains physical and information delays are random which means that a member of a supply chain does not know values of the future lead times and in the past he observed that their values were varied in a stochastic manner. For instance in the paper of So and Zheng \cite{so:zh:03} the model of a supply chain with stochastic lead times is motivated by the semiconductor industry where the dramatic boom-and-bust cycles cause the delivery lead times to be highly variable, ranging from several weeks during the low demand season to over several months during the high demand season. 
Moreover in the models investigating the bullwhip effect one can decide how time
is represented. There are two choices that is discrete or continuous time. 
We analyze stochastic techniques which use discrete time. We assume that the observations are made at integer moments of time which means that time is represented in units of the review periods and nothing  is known about the system in the time between observations.

The main difference in models with stochastic lead times lies in the definition of the lead time demand forecast. Let us recall that the lead time demand at the beginning of a period $t$ (at a certain stage of the supply chain) is defined as follows
\begin{equation}\label{ltd}
D_t^L=D_t+D_{t+1}+\ldots+D_{t+L_t-1}=
\sum_{i=0}^{L_t-1}D_{t+i}
\end{equation}
where $D_t, D_{t+1}, \ldots$ denote demands (from a stage below) during $t, t+1,\ldots$ periods and $L_t$ is the lead time
of the order placed at the beginning of the period $t$ (order placed to a stage above).
This value sets down the demand during a lead time. The demands come from the stage right below and the lead times come from the supplier right above that is they are delivery lead times of the supplier which is right above the receiving supply chain member. This quantity is necessary to place an
order. Practically the member of the supply chain does not know its value at the moment $t$ but he needs to predict its value to place an order. Thus if we want to analyze the bullwhip effect we need to look closer at the definitions of lead time demand forecasting $\widehat{D}_t^L$. The approaches to this problem vary greatly in models with stochastic lead times
and some of them cannot be feasible in practice.
Many articles on the bullwhip effect investigate different methods of demand forecasting under the assumption that lead times are constant. The problem
of the lead time demand prediction is much more complicated if lead times are stochastic. Then mere demand forecasting is not sufficient to place an order. 

We will analyze the works which quantify the bullwhip effect in supply chains with stochastic lead times.
In all the presented models we will consider a simple two stage supply chain consisting of customers, a retailer and
a manufacturer. Moreover will assume that the retailer uses the order-up-to-level policy (which is optimal in the sense that it minimizes the total discounted linear holding and backorder costs) then
the level of the inventory at time $t$ has to be
\begin{equation}\label{st}
S_t=\widehat{D}_t^L+z\widehat{\sigma}_t\,,
\end{equation}
where $\widehat{D}_t^L$ is the lead time demand forecast at the beginning of the period $t$
(that is the prediction of the quantity given in (\ref{ltd})) and
\begin{equation}\label{sigmadl}
\widehat{\sigma}_t^2=\Var(D_t^L-\widehat{D}_t^L)
\end{equation}
is the variance of the forecast error for the lead time demand
and $z$ is the normal z-score that specifies the probability that demand is fulfilled by the on-hand inventory and it can be found based on a given service level. Usually $z=\Phi(p/(p+h))$ where $\Phi$
is the standard normal cumulative distribution function and $h$ and $p$ are the unit inventory holding
and backorder costs at the retailer, respectively. Moreover we need to notice that the definition of  $\widehat{\sigma}_t^2$ differs in articles
(see e.g. Chen et al. \cite{ch:dr:ry:si:00a}, Duc et al. \cite{du:lu:ki:08} or Kim et al. \cite{ki:ch:ha:ha:06}) which results in slightly different formulas of the bullwhip effect measure (e.g. equality instead of inequality). Practically instead of variance we have to put the empirical variance of  $D_t^L-\widehat{D}_t^L$. This complicates theoretical calculations very much but we must mention that the estimation of  $\widehat{\sigma_t}^2$ increases the size of the bullwhip effect.
Thus under above assumptions the order quantity $q_t$ placed by the retailer
at the beginning of a period $t$ is
\begin{equation}\label{qt}
q_t=S_t-S_{t-1}+D_{t-1}\,.
\end{equation}
Let us notice that negative values of $q_t$ are allowed which correspond to returns.

\subsection{Lead time demand forecasting using moving average.}\label{kimchat}
Let us analyze the work of Kim et al. \cite{ki:ch:ha:ha:06} (see also Chatfield \cite{ch:ki:ha:ha:04}
for a simulation approach). In their approach lead time demand forecasting is defined as follows
$$
\widehat{D}_t^L=\frac{1}{n}\sum_{j=1}^{n}D_{t-j}^L\,,
$$ 
where $n$ is the delay parameter of the prediction and $D_{t-j}^L$ is the previous known lead time demand of the order placed at the beginning of the time $t-j$. This method is practically feasible. 
The problem of the approach of Kim et al. \cite{ki:ch:ha:ha:06} lies in an impractical definition of the past lead time demands $D_{t-j}^L$. Namely they continue
$$
\widehat{D}_t^L=\frac{1}{n}\sum_{j=1}^{n}D_{t-j}^L=
\frac{1}{n}\sum_{j=1}^{n}\sum_{i=0}^{L-1} D_{t-j+i}=
\frac{1}{n}\sum_{i=0}^{L-1}\sum_{j=1}^n D_{t-j+i}\,,
$$
where $L$ is a lead time. Firstly if we assume that lead times are stochastic then with every lead time
demand $D_{t-j}^L$ we associate a different lead time $L_{t-j}$. Moreover this definition does not
work in the case of a deterministic lead time because at the beginning of the moment $t$ the values of demands $D_{t-j+i}$ if $j\leq i$ are not known (they explain it is a ''mirror image'' and ''equivalent in terms of a priori statistical analysis'' ).
%In the case of iid demands the above assumptions cannot affect the value of the bullwhip effect %measure.

Let us analyze the bullwhip effect under above setting but with small modifications (see also Michna et al. \cite{mi:ni:ni:13}). More precisely let us consider the simplest supply chain that consists of customers, a retailer and a supplier. We assume that the customer demands constitute an iid sequence $\{D_t\}_{t=-\infty}^\infty$. Moreover lead times are deterministic and equal to $L$ where $L$
is a positive integer that is $L=1,2,\ldots$. It is assumed that the retailer's replenishment order policy is  
the order-up-to-level policy and his lead time demand forecasting is based on the moving average method. 
Thus the forecast of the lead time demand at the beginning of the period $t$ based on the moving average method is as follows
\begin{equation}\label{ltdfma}
\widehat{D}_t^L=\frac{1}{n}\sum_{i=0}^{n-1} D^L_{t-L-i}\,.
\end{equation} 
Let us notice that we have to get back with lead time demands at least to the period $t-L$ because
we know demands till the epoch $t-1$. 
Moreover let us recall that the demand forecast alone using the moving average is as follows
$$
\widehat{D}_t=\frac{1}{n}\sum_{j=1}^n D_{t-j}\,.
$$
Thus substituting into eq. (\ref{ltdfma}) the known values of the previous lead time demands we get
\begin{eqnarray}
\widehat{D}_t^L&=&\frac{1}{n}\sum_{i=0}^{n-1}\sum_{j=0}^{L-1}D_{t-L-i+j}\nonumber\\
&=&\sum_{j=0}^{L-1}\frac{1}{n}\sum_{i=0}^{n-1}D_{t-L-i+j}\nonumber\\
&=&\sum_{j=0}^{L-1}\widehat{D}_{t-L+j+1}\nonumber\\
&=&\sum_{j=0}^{L-1}\widehat{D}_{t-j}\,.\label{ltdfkim}
\end{eqnarray}
Applying the order-up-to-level policy we get that the inventory level of the retailer at time $t$ is given in (\ref{st}).
The error of the lead time demand forecast$\widehat{\sigma}_t$ is defined as in eq. (\ref{sigmadl}).
It is easy to notice that under above assumptions $\widehat{\sigma}_t$ is independent of $t$.
Thus the order quantity $q_t$ placed by the retailer
at the beginning of a period $t$ is
\begin{eqnarray*}
q_t&=&S_t-S_{t-1}+D_{t-1}\\
&=&\widehat{D}_t^L-\widehat{D}_{t-1}^L+D_{t-1}\\
&=&\sum_{j=0}^{L-1}\widehat{D}_{t-j}-\sum_{j=0}^{L-1}\widehat{D}_{t-1-j}+D_{t-1}\\
&=& \widehat{D}_t-\widehat{D}_{t-L}+D_{t-1}
\end{eqnarray*}
where in the second last equality we use eq. (\ref{ltdfkim}) .

To calculate the value of $q_t$  we need to consider two cases that is $L\geq n$ and $L<n$.
Thus in the case $L\geq n$ the order $q_t$ placed by the retailer is as follows
$$
q_t=\left(\frac{1}{n}+1\right)D_{t-1}+\frac{1}{n}\sum_{j=2}^n D_{t-j}
-\frac{1}{n}\sum_{j=1}^n D_{t-L-j}
$$
and
\begin{eqnarray*}
\Var q_t&=&\left[\left(\frac{1}{n}+1\right)^2+\frac{2n-1}{n^2}\right]\Var D\\
&=&\left(1+\frac{4}{n}\right)\Var D\,.
\end{eqnarray*}
In the case $L<n$ we get
$$
q_t=\left(\frac{1}{n}+1\right)D_{t-1}+\frac{1}{n}\sum_{j=2}^L D_{t-j}
-\frac{1}{n}\sum_{j=1}^L D_{t-p-j}
$$
and
\begin{eqnarray*}
\Var q_t&=&\left[\left(\frac{1}{n}+1\right)^2+\frac{2L-1}{n^2}\right]\Var D\\
&=&\left(1+\frac{2}{n}+\frac{2L}{n^2}\right)\Var D\,.
\end{eqnarray*}
\begin{proposition}\label{ldet}
If the lead times are deterministic and positive integer valued that is $L=1,2,\ldots$ and lead time
demands are forecasted using the moving average method then bullwhip effect measure is
$$
BM=\frac{\Var q_t}{\Var D}=
\left\{
\begin{array}{ll}
1+\frac{2}{n}+\frac{2L}{n^2}& \mbox{if}\,\,\, L<n\\
1+\frac{4}{n}& \mbox{if}\,\,\, L\geq n\,.
\end{array}
\right.
$$
\end{proposition}
In Fig. \ref{det} we plotted the bulwhip effect measure for a deterministic lead time when 
lead time demands are predicted by the moving average method (see Prop. \ref{ldet}). 
Let us notice that the bullwhip effect function $BM(n)$ as a function of $n$ does not have any jump at
$L$ that is it smoothly gets across the point $n=L$ (compare Prop. \ref{ldet} with the similar result of Kim et al. \cite{ki:ch:ha:ha:06}).
\begin{figure}[!h]
\begin{center}
\includegraphics[height=10cm,width=10cm]{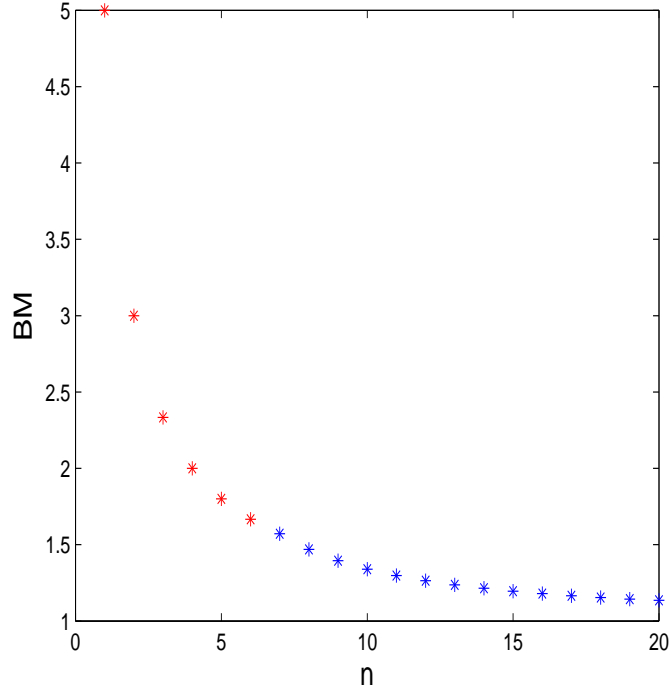}
\caption{The plot of the bullwhip effect measure as a function of $n$ where $L=7$.}\label{det}
\end{center}
\end{figure}

Now we follow the work of Kim et al. \cite{ki:ch:ha:ha:06} with certain modifications to find the bullwhip effect measure in the presence of stochastic lead times (see also Michna et al. \cite{mi:ni:ni:13}). We assume that the customers demands constitute an iid sequence $\{D_t\}_{t=-\infty}^\infty$ and the lead times $\{L_t\}_{t=-\infty}^\infty$
are also independent and identically distributed and the sequences are mutually independent.
Let us put $\Exp D_t=\mu_D$, $\Var D_t=\sigma^2_D$, $\Exp L_t=\mu_L$ and 
$\Var L_t=\sigma^2_L$.
Additionally we need to assume that lead times are bounded random variables that is $L_i\leq M$
where $M$ is a positive integer. This assumption is not adopted in Kim et al. \cite{ki:ch:ha:ha:06} which makes their results slightly impractical because it is necessary to make the prediction of lead time demands.
More precisely we get back at least $M$ periods to forecast lead time demand that is at time $t$ we know lead time demands of times $t-M, t-M-1, \ldots$ and we may not know lead time demands of times $t-M+1, t-M+2, \ldots$.
As we see later we will need to know the distribution of $L_t$ to calculate the bullwhip effect measure that is 
we assume that 
$$
\Prob(L_t=k)=p_k
$$
where $k=1,2,\ldots,M$ and $k$ is the number of periods (in practice we estimate these probabilities).
Thus the prediction of the lead time demand at time $t$ using the method of moving average with the length $n$ is as follows
\begin{equation}\label{ltdp}
\widehat{D}_t^L=\frac{1}{n}\sum_{j=0}^{n-1} D_{t-M-j}^L\,.
\end{equation}
Once again let us notice that the lead time demands $D_{t-M+1}^L, D_{t-M+2}^L, \ldots$
we may not know at time $t$ that is why in the lead time demand forecasting we engage $D_{t-M}^L, D_{t-M-1}^L, \ldots$ which are the lead time demands up to time $t-M$. Thus by eq. (\ref{ltd}) we get
\begin{equation}\label{ltdpc}
\widehat{D}_t^L=\frac{1}{n}\sum_{j=0}^{n-1} \sum_{i=0}^{L_{t-M-j}-1}D_{t-M-j+i}\,.
\end{equation}
As before the retailer uses the order-up-to-level policy thus
the level of the inventory at time $t$ is given in eq. (\ref{st}). By the stationarity and independence of the sequences of demands and lead times one can show that $\widehat{\sigma}_t^2$ given in (\ref{sigmadl}) does not depend on $t$. Hence we obtain
\begin{eqnarray}
q_t&=& \widehat{D}_t^L-\widehat{D}_{t-1}^L+D_{t-1}\nonumber\\
&=&\frac{1}{n}\sum_{j=0}^{n-1} D_{t-M-j}^L-\frac{1}{n}\sum_{j=0}^{n-1} D_{t-1-M-j}^L +D_{t-1}\nonumber\\
&=& \frac{1}{n} D_{t-M}^L-\frac{1}{n}D_{t-M-n}^L+D_{t-1}\nonumber\\
&=&\frac{1}{n}\sum_{i=0}^{L_{t-M}-1}D_{t-M+i}-\frac{1}{n}\sum_{i=0}^{L_{t-M-n}-1}
D_{t-M-n+i}+D_{t-1}\,.\label{qtcal}
\end{eqnarray}
\begin{theorem}\label{thltdp1}
Under above assumptions and for $n\geq M$ the bullwhip effect measure is the following
$$
BM=\frac{\Var q_t}{\Var D_t}=
1+\frac{2p_M}{n}+\frac{2\mu_L}{n^2}+\frac{2\mu^2_D\sigma^2_L}{\sigma^2_D n^2}\,.
$$
\end{theorem}
\begin{proof}
Using the law of total variance we have
$$
\Var q_t=\Exp(\Var(q_t|L_{t-M}, L_{t-M-n}))+\Var\Exp(q_t|L_{t-M}, L_{t-M-n})\,.
$$
By eq. (\ref{qtcal}) we get
$$
\Exp(q_t|L_{t-M}, L_{t-M-n})=\mu_D\left(\frac{L_{t-M}-L_{t-M-n}}{n}+1\right)\,.
$$
Thus
\begin{equation}\label{varexp}
\Var\Exp(q_t|L_{t-M}, L_{t-M-n})=\frac{2\mu^2_D\sigma^2_L}{n^2}\,.
\end{equation}
We need to consider two case to find $\Var(q_t|L_{t-M}, L_{t-M-n})$.
In the first case $L_{t-M}<M$ we get
$$
\Var(q_t|L_{t-M}, L_{t-M-n})=\sigma^2_D\left(\frac{L_{t-M}+L_{t-M-n}}{n^2}+1\right)\,.
$$
If $L_{t-M}=M$ we have
$$
\Var(q_t|L_{t-M}, L_{t-M-n})=\sigma^2_D\left(\frac{L_{t-M}+L_{t-M-n}}{n^2}+1+\frac{2}{n}\right)\,.
$$
Finally we obtain
$$
\Var(q_t|L_{t-M}, L_{t-M-n})=\sigma^2_D\left(\frac{L_{t-M}+L_{t-M-n}}{n^2}+1\right)+
\frac{2\sigma^2_D}{n}\ind\{L_{t-M}=M\}
$$
where $\ind$ is the indicator function. Thus we get
$$
\Exp\Var(q_t|L_{t-M}, L_{t-M-n})=\sigma^2_D\left(\frac{2\mu_L}{n^2}+1\right)+\frac{2\sigma^2_D p_M}{n}
$$
which together with eq. (\ref{varexp}) give the assertion.
\end{proof}

The formula for the bullwhip effect measure in the case $n<M$ is more complicated and its derivation is rather cumbersome. In practice the case $n\geq M$ is more interesting because we require to put large $n$ in the forecast to get a more precise prediction. Let us notice that if $L_t=M=L$ is deterministic the formula of Th. \ref{thltdp1} is consistent with Prop. \ref{ldet}.

We have to mention that in the formula of Th. \ref{thltdp1} the term $\frac{2p_M}{n}$ gives the largest contribution in the bullwhip effect  for large $n$ because it is of the order $O(1/n)$. This means that in reducing the bullwhip effect the probability of the largest lead time is very important. It is astonishing that if $p_M=0$ and we still get back $M$ periods in the prediction of lead time demands then the bullwhip effect measure is reduced
by the therm $O(1/n)$ and is of the form
$$
\frac{\Var q_t}{\Var D_t}=
1+\frac{2\mu_L}{n^2}+\frac{2\mu^2_D\sigma^2_L}{\sigma^2_D n^2}\,.
$$
Let us compare the values of the bullwhip effect measure under $p_M>0$ and $p_M=0$. More precisely let $L_t$ have the discrete uniform distribution on $\{1,2,3\}$ that is 
$p_k=1/3$ for $k=1,2,3$ then $M=3$, $\mu_L=2$ and $\sigma_L^2=2/3$. In the case $p_M=0$ we assume that $L_t$ has the discrete uniform distribution on $\{1,2\}$ that is 
$p_k=1/2$ for $k=1,2$ then $\mu_L=1.5$ and $\sigma_L^2=1/4$ (and we still get back
at least $M=3$ periods to predict the lead time demand). The results are in Tab. \ref{bm1}.
\begin{table}[!h]
  \begin{center}
  \caption{The measure of the bullwhip effect as a function of $n$ for $M=3$ and $\sigma_D/\mu_D=0.5$.}\label{bm1}
\vspace{2mm}
   \begin{tabular}{|c|c|c|}
     \hline
$n$ & $p_M>0$ & $p_M=0$ 
\\ \hline
3&	2.259&	1.555\\
4&	1.750&	1.312\\
5&	1.506 &	1.200\\
6&	1.370&	1.138\\
7&	1.285& 1.102\\
8&	1.229& 1.078\\
9&	1.189&	1.061\\
10&	1.160&	1.050\\
11&	1.137&	1.041\\
12&	1.120&	1.034\\
13&	1.106&	1.029\\
14&	1.095	& 1.025\\
15&	1.085& 	1.022\\
\hline
\end{tabular}
\end{center}
\end{table}
Similarly we can calculate the bullwhip effect measure for longer lead times. More precisely let $L_t$ have the discrete uniform distribution on $\{1,2,\ldots,7\}$ that is 
$p_k=1/7$ for $k=1,2,\ldots, 7$ then $M=7$, $\mu_L=4$ and $\sigma_L^2=4$. In the case $p_M=0$ we assume that $L_t$ has the discrete uniform distribution on $\{1,2,\ldots,6\}$ that is 
$p_k=1/6$ for $k=1,2,\ldots, 6$ then $\mu_L=3.5$ and $\sigma_L^2=3.916$ (and we still get back
at least $M=7$ periods to predict the lead time demand). The results are in Tab. \ref{bm2}.
\begin{table}[!h]
  \begin{center}
  \caption{The measure of the bullwhip effect as a function of $n$ for $M=7$ and $\sigma_D/\mu_D=0.5$.}\label{bm2}
\vspace{2mm}
   \begin{tabular}{|c|c|c|}
     \hline
$n$ & $p_M>0$ & $p_M=0$ 
\\ \hline
7&	1.857	 & 1.782\\
8&	1.660	 & 1.598\\
9&	1.525	 & 1.473\\
10& 1.428 &  1.383\\
11&1.356   & 1.316\\
12&1.301 & 1.266\\
13&1.258 & 1.226\\
14&1.224 & 1.195\\
15&1.196 & 1.170\\
16&1.174 & 1.149\\
17&1.155 &1.132\\
18&1.139& 1.118\\
\hline
\end{tabular}
\end{center}
\end{table}

\subsection{Stochastic lead times without forecasting}
In the paper of Duc et al. \cite{du:lu:ki:08} stochastic lead times are investigated under the assumption
that they are independent and identically distributed. The simplest  two stage supply chain in analyzed with a first-order autoregressive AR(1) demand process and an extension to a mixed first-order autoregressive-moving average ARMA(1,1). More precisely
the demands from customers to the retailer
constitute the first order autoregressive-moving average AR(1) that is $\{D_t\}_{t=-\infty}^\infty$ is a stationary sequence of random variables which satisfy
\begin{equation}\label{ard} 
D_t=\mu+\rho D_{t-1}+\epsilon_t
\end{equation}
where $\mu>0$, $|\rho|<1$ and $\{\epsilon_t\}_{t=-\infty}^\infty$ is a sequence of independent   
identically distributed random variables such that $\Exp \epsilon_t=0$ and $\Var\epsilon_t=\sigma^2$.
It is easy to notice that $\Exp D_t=\mu_D=\frac{\mu}{1-\rho}$, $\Var D_t=\sigma^2_D=\frac{\sigma^2}{1-\rho^2}$ and the correlation coefficient $\Corr(D_t, D_{t+1})=\rho$

Moreover it is assumed that the demands are forecasted using the minimum-mean-squared-error forecasting method. If $\widehat{D}_{t+i}$ denotes the forecast for a demand for the period $t+i$ at the beginning of a period $t$ (that is after $i$ periods) then employing the minimum-mean-squared-error forecasting method we get
\begin{eqnarray}\label{msef}
\widehat{D}_{t+i}&=&\Exp (D_{t+i}|, D_{t-1}, D_{t-2},\dots)\nonumber\\
&=& \mu_D(1-\rho^{i+1})+\rho^{i+1}D_{t-1}\label{fd}
\end{eqnarray}
where $D_{t-j}$ $j=1,2,\ldots$ are demands which have been observed by the retailer till the
beginning of a period $t$. Then the lead time demand at the beginning of the period $t$ is defined by
 Duc et al. \cite{du:lu:ki:08}
as follows
$$
\widehat{D}_t^L=\sum_{i=0}^{L_t-1}\widehat{D}_{t+i}
$$
where $\widehat{D}_{t+i}$ is given in eq. (\ref{msef}). 
Let us notice that the above lead time demand forecast is not practically feasible because we do not know the value of $L_t$ and the beginning of the period $t$. Practically to place an order we have to forecast demands and lead times which means that in the above lead time demand forecast we need to substitute a lead time prediction $\widehat{L}_t$ instead of $L_t$. 
As in the previous model the retailer uses the order-up-to-level policy and the level of inventory $S_t$ is given in (\ref{st}). The variance of the forecast error for the lead time demand and the order quantity $q_t$ placed by the retailer
at the beginning of a period $t$ are defined in (\ref{sigmadl}) and (\ref{qt}), respectively. 
The main result of Duc et al. \cite{du:lu:ki:08} is the following.
\begin{theorem}\label{duc1}
Under the above assumptions with the minimum-mean-squared-error forecasting method the bullwhip effect measure is 
\begin{eqnarray*}
\lefteqn{BM=\frac{\Var q_t}{\Var D_t}}\\
&=&\frac{1}{(1-\rho)^2}[(1-\rho^2)[1-2\rho\Exp\rho^{L_t}]+2\rho^2\Exp\rho^{2 L_t}
-2\rho^3(\Exp\rho^{L_t})^2]+\frac{2\mu^2_D\sigma^2_L}{\sigma^2_D}\,.
\end{eqnarray*}
\end{theorem}
Duc et al. \cite{du:lu:ki:08} give numerical examples. They calculate the value of $BM$ for specific distributions of $L_t$ e.g. three-point distribution, geometric distribution,
Poisson and discrete uniform distribution. The plots of $BM$ as a function of the autoregressive coefficient $\rho$ for a fixed $\sigma_D/\mu_D$ are presented. It is interesting that the minimal value of $BM$ is attained for $\rho$ around $-0.6$ and $-0.7$. The maximal value of $BM$ is for $\rho$ around
$0.6$ or $1$.

The authors of \cite{du:lu:ki:08} extend the results for ARMA(1,1) demand processes (the mixed first-order autoregressive-moving average process). In this case, the structure of demands is defined as follows
$$
D_t=\mu+\rho D_{t-1}+\epsilon_t-\theta\epsilon_{t-1}
$$
where $\mu$, $\rho$ and $\epsilon_t$ are the same as in the case of AR(1) demand process and $|\theta|<1$. Then
under the same assumptions (the order-up-to-level inventory policy and the minimum-mean-squared-error forecasting method) the bullwhip effect measure is given.
\begin{theorem}\label{duc2}
Under ARMA(1,1) demand process with the minimum-mean-squared-error forecasting method the bullwhip effect measure is
\begin{eqnarray*}
\lefteqn{BM=\frac{\Var q_t}{\Var D_t}}\\
&=&\frac{(1-\rho^2)(1-\theta)[1-\theta+2(\theta-\rho)\Exp\rho^{L_t}]+2(\rho-\theta)^2[\Exp\rho^{2 L_t}-\rho(\Exp\rho^{L_t})^2]}{(1-\rho)^2 (1+\theta^2-2\rho\theta)}\\
&&\,\,\,\,\,\,\,\,\,\,\,\,\,\,\,\,\,\,\,\,\,\,\,\,\,\,\,\,\,\,\,\,\,\,\,\,\,\,
+\frac{2\mu^2_D\sigma^2_L}{\sigma^2_D}\,.
\end{eqnarray*}
\end{theorem}
Numerical results for the case of ARMA(1,1) demand
process provide the same trends as those of the AR(1) case.

\subsection{Stochastic lead times with forecasting}
In the work of Michna and Nielsen \cite{mi:ni:13} the impact of lead time forecasting on the bullwhip effect is investigated. It is assumed that lead times and demands are forecasted separately which seems to be a very natural and practical approach. More precisely the lead time demand prediction is the following 
\begin{equation}\label{eltd}
\widehat{D}_t^L=\widehat{L}_t\widehat{D}_t=
\frac{1}{mn}\sum_{i=1}^m L_{t-i}\sum_{i=1}^n D_{t-i}\,.
\end{equation}
where we use the moving average method for lead times and demands with the delay parameters 
$m$ and $n$, respectively. Moreover we assume that lead times and demands constitute  i.i.d. sequences being mutually independent. Under the same assumptions as in the previous models on the policy and the lead time demand forecast error it is proven the following result.
\begin{theorem}\label{bmmt}
The measure of the bullwhip effect in a two stage supply chain has the following form
$$
BM=\frac{\Var q_t}{\Var D_t}=
\frac{2\sigma_L^2(m+n-1)}{m^2n^2}+\frac{2\mu_D^2\sigma_L^2}{\sigma_D^2m^2}+\frac{2\mu_L^2}{n^2}+\frac{2\mu_L}{n}+1\,.
$$
\end{theorem}
The above theoretical model shows that one cannot avoid lead time forecasting when placing orders and the variance of orders will increase dramatically if a crude estimation of lead time (e.g. small $m$) or no estimation is used (e.g. assuming a constant lead time when placing orders).  Moreover the demand signal processing and the lead time signal processing which mean the practice of adjusting demand and lead time forecasts resulting in adjusting the parameters of the inventory replenishment rule are the main and equally important causes of the bullwhip effect.

To confirm theoretical results derived in Michna and Nielesen  \cite{mi:ni:13} we simulate the bullwhip effect measure in a supply chain which consists of three echelons.
First we assume that client demands are deterministic that is during a given period (this will be a week) we observe the same constant demand  $D$.
Above the consumers in our supply chain we have a retailer, a manufacturer and a supplier. Between the manufacturer and the retailer there are stochastic lead times which create an i.i.d. sequence (that is they are the delivery times of the manufacturer to the retailer). Similarly we observe random times between the supplier and the manufacturer and they constitute an i.i.d. sequence. These two sequences are mutually independent. Moreover we take the review period equal to
a week (7 days) and the lead times are uniform discrete random variables taking on values $1,2,\ldots,7$. The retailer uses the order-up-to level policy and the moving average method to predict
lead times with the delay parameter $m$ (consumer demands are constant equal to 5000 so they are not predicted by the retailer). Similarly the manufacturer places orders to its supplier that is he uses the order-up-to-level policy and the moving average method to predict
lead times with the delay parameter $m$ and the demands with the delay parameter $n$ (the demands of
the retailer are random by random lead times in his lead time demand forecast). In Tab. \ref{bmmdet} there are the simulation results that is the ratio of variances of the manufacturer and the retailer orders (variance of consumer demands is zero). A common feature of these simulation results is the fact that the delay parameter of demand forecasting $n$ smooths bullwhip much faster than the delay parameter of lead time forecasting $m$.

Under the same assumptions as above we simulate the bullwhip effect adding that customer demands are stochastic and they are i.i.d. with uniform distribution on $(4500, 5500)$ and independent of lead times. In Tab. \ref{bmmsto1} the bullwhip effect at the retailer stage is given that is the quotient of the 
retailer variance and the customer demand variance. Tab. \ref{bmmsto2} shows the same as in Tab.
\ref{bmmsto1} but calculated theoretically using the formula of Th. \ref{bmmt}.
Here we get a reverse behavior than in the case of deterministic demands that is the delay parameter of lead time forecasting $m$ smooths bullwhip much faster than the delay parameter of demand forecasting $n$
(see Tab. \ref{bmmsto1} and   Tab. \ref{bmmsto2}). 
Finally in Tab. \ref{bmmsto3} we have the bullhwip effect measure at the manufacturer stage 
that  is the ratio of the manufacturer order variance and the customer demand variance (we could count  the quotient of the manufacturer order variance and the retailer order variance but it is easy to get this having also the ratio of the 
retailer variance and the customer demand variance).
The simulation results for the bullwhip effect at the manufacturer stage 
show that the delay parameter of demand forecasting $n$ and the delay parameter of lead time forecasting $m$ dampen the effect with a similar strength.
\begin{table}[!h]
  \begin{center}
  \caption{The bullwhip effect measure for discrete uniform lead times and constant customer demands}\label{bmmdet}
\vspace{2mm}
 \begin{tabular}{|c|c|c|c|c|c|}
     \hline
$m \backslash n$  & 1 &2  & 6 & 10 & 20
\\ \hline
1 & 61.6592 & 17.9880 & 4.4984 & 3.3532 & 2.4048\\
3 & 39.1578&14.5778 & 4.3510& 3.1641 & 2.4692\\
6 & 44.2991&14.4909 & 5.4784 &3.0812 & 2.4720 \\
10 & 42.9382 &13.8638 & 4.3274& 3.6823 & 2.4074\\
15 &42.4075 &14.5218 & 4.0920& 3.1734 & 2.5155\\
20 &43.292 & 14.194 & 4.150& 3.165 & 2.744\\
\hline
\end{tabular}
\end{center}
\end{table}

\begin{table}[!h]
  \begin{center}
  \caption{The bullwhip effect measure at the retailer stage for discrete uniform lead times and uniformly distributed customer demands}\label{bmmsto1}
\vspace{2mm}
\begin{tabular}{|c|c|c|c|c|c|}
     \hline
$m \backslash n$  & 1 &2  & 6 & 10 & 20
\\ \hline
1 & 2506.6 &2336.7  &  2392.4&2380.9 &2549.4 \\
3 &310.13 &  280.31& 269.90 &267.19  & 265.30\\
6 &107.60 &78.70  & 68.65 &  71.21& 68.86\\
10 & 67.313 & 37.270 & 26.465 & 25.997 &25.889 \\
20 & 46.8218 &19.3528  &  9.2602& 8.1362 & 7.4563\\
\hline
\end{tabular}
\end{center}
\end{table}

\begin{table}[!h]
  \begin{center}
  \caption{The bullwhip effect measure at the retailer stage for discrete uniform lead times and uniformly distributed customer demands calculated theoretically}\label{bmmsto2}
\vspace{2mm}
\begin{tabular}{|c|c|c|c|c|c|}
     \hline
$m \backslash n$  & 1 &2  & 6 & 10 & 20
\\ \hline
1 & 2449.0 & 2417.0  & 2404.6 & 2402.9& 2401.9\\
3 & 310.33& 280.55 &  270.08& 268.89 &268.19 \\
6 & 109.00& 80.055 & 69.956 & 68.820 &68.160 \\
10 & 65.800&  37.220&  27.255&26.135  & 25.485\\
20 &47.400  & 19.105 & 9.2361 & 8.1258 & 7.4820\\
\hline
\end{tabular}
\end{center}
\end{table}

\begin{table}[!h]
  \begin{center}
  \caption{The bullwhip effect measure at the manufacturer stage for discrete uniform lead times and uniformly distributed customer demands}\label{bmmsto3}
\vspace{2mm}
\begin{tabular}{|c|c|c|c|c|c|}
     \hline
$m \backslash n$  & 1 &2  & 6 & 10 & 20
\\ \hline
1 &194700 &41742& 10720 & 8024.0 & 5891.9\\
3 &  13495& 4221.1 &  1191.6& 856.245 & 676.579\\
6 & 5821.1& 1216.1 & 364.729 & 226.268 &  171.459\\
10 &3671.5 & 581.341 & 112.412 & 93.529 &62.333 \\
20 & 2840.2 & 316.308 & 37.244 & 23.710 & 18.972\\
\hline
\end{tabular}
\end{center}
\end{table}

\section{Conclusions and future research opportunities}
The main conclusion from our research is that stochastic lead times boost the bullwhip effect. More precisely we deduce from the presented models that the effect is amplified by the increase of the expected value, variance and the probability of the largest lead time. Moreover the delay parameter of the prediction
of demands, the delay parameter of the prediction of lead times and the delay parameter of the prediction of lead time demands depending on the model are crucial parameters which can dampen the bullwhip effect. We must also notice that in all the presented models the bullwhip effect
measure contains the term  $\frac{2\mu_D^2\sigma_L^2}{\sigma_D^2}$ (see Th. \ref{thltdp1}, \ref{duc1}, \ref{duc2} and \ref{bmmt}) and except the model of Duc et al. \cite{du:lu:ki:08} this
term can be killed by the prediction (going with $n$ or $m$ to $\infty$).

The future research on quantifying the bullwhip effect has to be aimed  at stochastic lead times with a different structure than i.i.d. and dependence between lead times and demands. One can investigate for example AR(1) structure of lead times and the influence of the dependence between lead times and demands on the bullwhip effect.
Another challenge in bullwhip modeling is the problem of lead time forecasting and its impact on the bullwhip effect. A member of a supply chain placing an order must forecast lead time to determine an appropriate inventory level in order to fulfill its customer orders in a timely manner which implies that lead times influence orders. In turn orders can impact lead times. This feedback loop can be the most important factor causing the bullwhip effect which has to be quantified
and in our opinion this seems to be the most important challenge and the most difficult problem in bullwhip modeling. 
Another topic is the combination of methods for lead time forecasting and demand forecasting (to predict lead time demand). Thus the spectrum of models which have to be investigated in order to quantify and find all causes of the bullwhip effect is very wide. However, these problems do not seem to be easy to solve by providing analytical models alone.

\subsection*{Acknowledgments} This work has been supported by the National Science Centre grant \\ 2012/07/B//HS4/00702.


\begin{thebibliography}{99}

\bibitem{ag:se:sh:09} 
Agrawal, S., Sengupta, R.N., Shanker, K. (2009).
Impact of information sharing and lead time on bullwhip effect and
on-hand inventory. {\em European Journal of Operational Research} {\bf 192}, pp. 576--593.

\bibitem{al:li:ya:03}
Alwan L.C., Liu J., J., Yao, D.Q. (2003). Stochastic characterization of upstream demand processes in a supply chain. {\em IIE Transactions} {\bf 35},
pp. 207--219.

\bibitem{ba:ha:ch:86}
Bagchi, U., Hayya, J., Chu, C. (1986). The effect of leadtime
variability: The case of independent demand. {\em Journal of
Operations Management} {\bf 6}, pp. 159--177.

\bibitem{bo:je:76}
Box, G.E.P., Jenkins, J.M. (1976). {\em Time series analysis: Forecasting and control.} Holden-Day,
San Francisco.

\bibitem{bu:pa:pa:po:08}
Buchmeister, B., Pavlinjek, J., Palcic, I., Polajnar, A. (2008). Bullwhip effect problem in supply chains. {\em Advances in Production Engineering and Management} {\bf 3}(1), pp. 45--55.

\bibitem{ch:he:10}
Chaharsooghi, S.K.,  Heydari, J. (2010).
LT variance or LT mean reduction in supply chain management: Which one
has a higher impact on SC performance? {\em International Journal of Production Economics} {\bf 124}, pp. 475--481.

\bibitem{ch:ki:ha:ha:04}
Chatfield, D.C., Kim, J.G., Harrison, T.P., Hayya, J.C. (2004). The
bullwhip effect - Impact of stochastic lead time, information
quality, and information sharing: a simulation study. {\em Production and  Operations
Management} {\bf 13}(4), pp. 340--353.

\bibitem{ch:dr:ry:si:00a}
Chen, F., Drezner, Z., Ryan, J.K., Simchi-Levi, D. (2000a). Quantifying
the bullwhip effect in a simple supply chain. {\em Management Science} {\bf 46}(3), pp. 436--443.

\bibitem{ch:dr:ry:si:00b}
Chen, F., Drezner, Z., Ryan, J.K., Simchi-Levi, D. (2000b). The impact
of exponential smoothing forecasts on the bullwhip effect. {\em Naval
Research Logistics} {\bf 47}(4), pp. 269--286.

\bibitem{co:71}
Conover, W.J. (1971). {\em Practical Nonparametric Statistics.} New York, John Wiley \& Sons.

 
\bibitem{di:to:03}
Disney, S.M., Towill, D.R. (2003). On the bullwhip and inventory variance produced by an ordering policy. {\em Omega} {\bf 31}, pp. 157--167.

\bibitem{du:lu:ki:08}
Duc, T.T.H., Luong, H.T., Kim, Y.D. (2008). A measure of the bullwhip effect in supply chains with stochastic lead time. {\em The International Journal of Advanced Manufacturing Technology} {\bf 38}(11-12), pp. 1201--1212.

\bibitem{fi:fo:08}
Fioriolli, J.C.,  Fogliatto, F.S. (2008). A model to quantify the bullwhip effect in systems with stochastic demand and lead time. {\em Proceedings of the 2008 IEEE IEEM}, pp. 1098--1102. 

\bibitem{fo:58}
Forrester, J.W. (1958). Industrial dynamics - a major break-through
for decision-making. {\em Harvard Business Review} {\bf 36}(4), pp. 37--66.

\bibitem{ge:di:to:06}
Geary, S., Disney, S.M., Towill, D.R. (2006). On bullwhip in supply chains - historical review, present practice and expected future impact. {\em International Journal Production Economics} {\bf 101}, pp. 2--18.

\bibitem{ha:zi:95} Hariharan, R., Zipkin, P. (1995). Costumer-order information, leadtimes, and inventories. {\em Management Science} {\bf 41}, pp. 1599--1607.

\bibitem{le:pa:wh:97a}
Lee, H.L., Padmanabhan, V., Whang, S. (1997a). The bullwhip effect
in supply chains. {\em Sloan Management Review} {\bf 38}(3), pp. 93--102.

\bibitem{le:pa:wh:97b}
Lee, H.L., Padmanabhan, V., Whang, S. (1997b). Information
distortion in a supply chain: the bullwhip effect. {\em Management Science} {\bf 43}(3), pp. 546--558.

\bibitem{le:so:ta:00}
Lee, H.L., So, K.C., Tang, C.S. (2000). The value of information sharing
in a two-level supply chain. {\em Management Science} {\bf 46}(5), pp. 626--643.

\bibitem{li:li:13}
Li, C., Liu, S. (2013).
A robust optimization approach to reduce the bullwhip effect of supply chains with vendor order placement lead time delays in an uncertain environment. {\em Applied Mathematical Modeling}
{\bf 37}, pp. 707--718.


\bibitem{gr:99}
Graves, S.C. (1999). A single-item inventory model for a non-stationary
demand process. {\em Manufacturing and Service Operations Management} {\bf 1}, pp. 50--61.

\bibitem{ha:zi:95}
Hariharan, R., Zipkin, P. (1995). Customer-order information,
leadtimes, and inventories. {\em Management Science} {\bf 41},
pp. 1599--1607.

\bibitem{ki:ch:ha:ha:06}
Kim, J.G., Chatfield, D., Harrison, T.P., Hayya, J.C. (2006). Quantifying the bullwhip effect in a supply chain with stochastic lead time. {\em European Journal of Operational Research} {\bf 173}(2), pp. 617--636.

\bibitem{mi:ni:13}
Michna, Z., Nielsen, P. (2013).
The impact of lead time forecasting on the bullwhip effect.
arXiv preprint arXiv:1309.7374

\bibitem{mi:ni:ni:13}
Michna, Z, Nielsen, I.E., Nielsen, P. (2013). The bullwhip effect in supply chains with stochastic lead times. {\em Mathematical Economics} {\bf 9}, pp. 71--88.

\bibitem{mo:po:98}
Mohebbi, E., Posner, M.J.M. (1998). A continuous-review
system with lost sales and variable lead time. {\em Naval
Research Logistics} {\bf 45}, pp. 259--278.

\bibitem{mu:hu:ro:03}
Munson, C.L., Hu, J., Rosenblatt, M. (2003). Teaching the costs of uncoordinated supply chains.
{\em Interfaces} {\bf 33}, pp. 24--39.

\bibitem{ni:14}
Nielsen, P., Michna, Z., Do, N.A.D. (2014). An empirical investigation of lead time distributions. {\em IFIP Advances in Information and Communication Technology} {\bf  438} (part I),
 pp. 435--442.

\bibitem{re:fi:09}
Reiner, G.,  Fichtinger, J. (2009). 
Demand forecasting for supply processes in consideration of pricing
and market information. {\em International Journal of Production Economics} {\bf 118}, pp. 55--62.

\bibitem{sa:za:91} Sarker, D., Zangwill, W. (1991). Variance effects in cyclic production systems. 
{\em Management Science} {\bf 40}, pp. 603--613.

\bibitem{so:zh:03}
So, K.C., Zheng, X. (2003). Impact of supplier’s lead time and forecast demand updating on retailer’s order quantity variability in a two-level supply chain.
{\em International Journal of Production Economics} {\bf 86}, pp. 169--179.

\bibitem{so:94a}
Song, J. (1994). The effect of leadtime uncertainty in simple
stochastic inventory model. {\em Management Science} {\bf 40},
pp. 603--613.

\bibitem{so:94b}
Song, J. (1994). Understanding the leadtime effects in
stochastic inventory systems with discounted costs. {\em Operations
Research Letters} {\bf 15}, pp. 85--93.

\bibitem{so:zi:93}
Song, J., Zipkin, P. (1993). Inventory control in a fluctuation
demand environment. {\em Operations Research} {\bf 41},
pp. 351--370.

\bibitem{so:zi:96}
Song, J., Zipkin, P. (1996). Inventory control with information
about supply chain conditions. {\em Management Science} {\bf 42},
pp. 1409--1419.


\bibitem{zh:04}
Zhang, X. (2004). The impact of forecasting methods on the bullwhip effect. {\em International Journal of Production Economics} {\bf 88}, pp. 15--27.

\bibitem{zi:86}
Zipkin, P. (1986). Stochastic leadtimes in continuous-time
inventory models. {\em Naval Research Logistics Quarterly} {\bf 33},
pp. 763--774.

\bibitem{zi:00}
Zipkin, P.H. (2000). {\em Foundations of Inventory Management.} McGraw-Hill, New York.



\end{thebibliography}
\end{document}